\documentclass{article}

\usepackage{amsmath, amsthm}    % need for subequations
\usepackage{graphicx}   % need for figures
\usepackage[utf8x]{inputenc}
\usepackage{latexsym,cmmib57}
\usepackage[psamsfonts]{amsfonts,eucal}
\usepackage{psfrag} % for pdf, bitmapped graphics files
\usepackage{epsfig} % for postscript graphics files
\usepackage{float}
\usepackage{graphics, subfigure}
\begin{document}

\title{The Exponentially Faster Stick-Slip Dynamics of the Peeling of an Adhesive Tape}
\author{Nachiketa Mishra \thanks{Supercomputer Education and Research Centre, Indian Institute of Science, Bangalore 560012, India.
Email: mishra.nachiketa@gmail.com}\and Nigam Chandra Parida \thanks{Supercomputer Education and Research Centre, Indian Institute of Science, Bangalore 560012, India.
Email:ncparida@gmail.com} \and Soumyendu Raha 
         \thanks{Supercomputer Education and Research Centre, Indian Institute of Science, Bangalore 560012, India. 
            Email:raha@serc.iisc.in}}
\date{\today}
\begin{abstract}
The stick-slip dynamics is considered from the nonlinear differential-algebraic equation (DAE) point of view and the peeling dynamics is 
shown to be a switching differential index DAE model. In the stick-slip regime with bifurcations, the differential
index can be arbitrarily high. The time scale of the peeling velocity, the algebraic variable, in this regime is
shown to be exponentially faster compared to the angular velocity of the spool and/or the stretch rate of the tape. 
A homogenization scheme for the peeling velocity which is characterized by the bifurcations is discussed
and is illustrated with numerical examples computed with the {\it $\alpha$-method}.
\end{abstract}
% \begin{keyword} 
% Differential-Algebraic Equation, Stick-Slip Dynamics, Bifurcation, Homogenization, Numerical Simulation of Stiff ODEs 
% \end{keyword}

\maketitle

\newtheorem{lem}{Lemma}
\newtheorem{thm}{Theorem}

\maketitle
\section{Introduction}
Experimental observations (see, for example, \cite{barquins1997kinetics, ciccotti13, cortet2007imaging} and references therein) via imaging and other
means have established the occurrence of bifurcations in the stick-slip dynamics of the peeling of a stretchable adhesive tape 
pulled with a constant velocity off a circular spool which is free to rotate around its center. The older theoretical studies
\cite{ciccotti1998stick, ciccotti2004complex} focused on the equations of motion and  the more recent ones 
\cite{de2006dynamics, kumar2008hidden, gandur1997complex, de2008lifting} recognized the differential-algebraic equation (DAE)
structure of the formulations, and proposed mathematical approaches such as the inertial regularization \cite{de2005missing}
to simulate the dynamics.
However, research issues such as (i) the differential index (see later in this work) of the DAE model, 
(ii) relationship of the differential index in the stick-slip regime to the bifurcation, and 
(iii) satisfaction of the adhesion-shear constitutive relationship (the algebraic constraint) in relation to the stick-slip regime of the dynamics
follow naturally from the DAE structure of the peeling dynamics. In this context it may be recalled that 
ordinary differential equation (ODE) analysis has its limitations in dealing with DAEs since DAEs are not ODEs \cite{petzold1982differential}.
\par
In this work, we address the above questions and show (i) that the peeling dynamics DAE is a variable differential index system, 
with differential index ranging from two to possibly very high values in the stick-slip regime compared to being one otherwise, 
(ii) that the higher differential index implies an exponentially faster time scale for the peeling velocity relative to the peel front angle between
the tape and the spool, and/or the tensile displacement of the tape, and
(iii) that the above two properties lead to the nonlinear bifurcations. Finally, (iv) we develop a homogenized ODE representation 
of the peeling dynamics DAE model along with a characteristic time scale for the homogenization. 

The homogenized ODE scheme is illustrated with physically consistent numerical simulations.

\section{The Peeling Model}
Following the modeling approach in \cite{de2004dynamics, ciccotti13}, the dynamics of the peeling of an adhesive tape off a circular spool is given as
\begin{subequations}
\begin{eqnarray}
&& L(\alpha):= \sqrt{R^2 + l^2 - 2R l \cos\alpha} \label{defL} \\
&& F(u, \alpha) := \frac{k u}{L(\alpha) - u}~~ \mbox{(tensile force in tape)} \label{defF} \\
&& \dot{\alpha} = \omega - \frac{v}{R}  \label{av} \\
&& \mbox{(rate of change of contact or peel front angle)} \nonumber \\
&& \dot{\omega} = - \frac{R l \sin\alpha}{I L(\alpha)}  F(u, \alpha) \label{tb} \\
&& \mbox{(torque balance for the spool)} \nonumber \\ 
&&\dot{u} = \frac{R l \sin \alpha}{L(\alpha)}\left(1 + \frac{u}{L(\alpha)}\right)\bigg( \omega - \frac{v}{R}\bigg) + V -v   \label{str} \\
&& \mbox{(rate of tensile displacement (stretch) of the tape)} \nonumber \\
&&\bigg(1+ \frac{l \sin\alpha}{L(\alpha)}\bigg) F(u,\alpha) - \phi (v,V)  = 0 \label{sac} \\
&& \mbox{(shear-adhesion constitutive relationship)} \nonumber 
\end{eqnarray}
\label{peel}
\end{subequations}
where $\alpha$, the contact or the peel front angle, is the angle between the radius passing through the point where the tape peels off the spool and
the radius coinciding with the horizontal, $u$ is the tensile displacement or stretch of the tape which has an Young's modulus 
$E$ and cross-sectional area $a$ so that $k:=Ea$, $v$ is the speed at which the tape is peeled at the point of the contact with the spool, or
simply the peeling velocity
and $V$ is the constant speed at which the free end of the tape is pulled along its longitudinal axis. The peeled
tape from the point of contact to the point where it is pulled, is assumed to have negligible mass compared to the circular spool which 
has only in-plane rotational degree of freedom about its center of mass. 
\begin{figure}[h,t,b]
 \centering
 \includegraphics[scale=0.3]{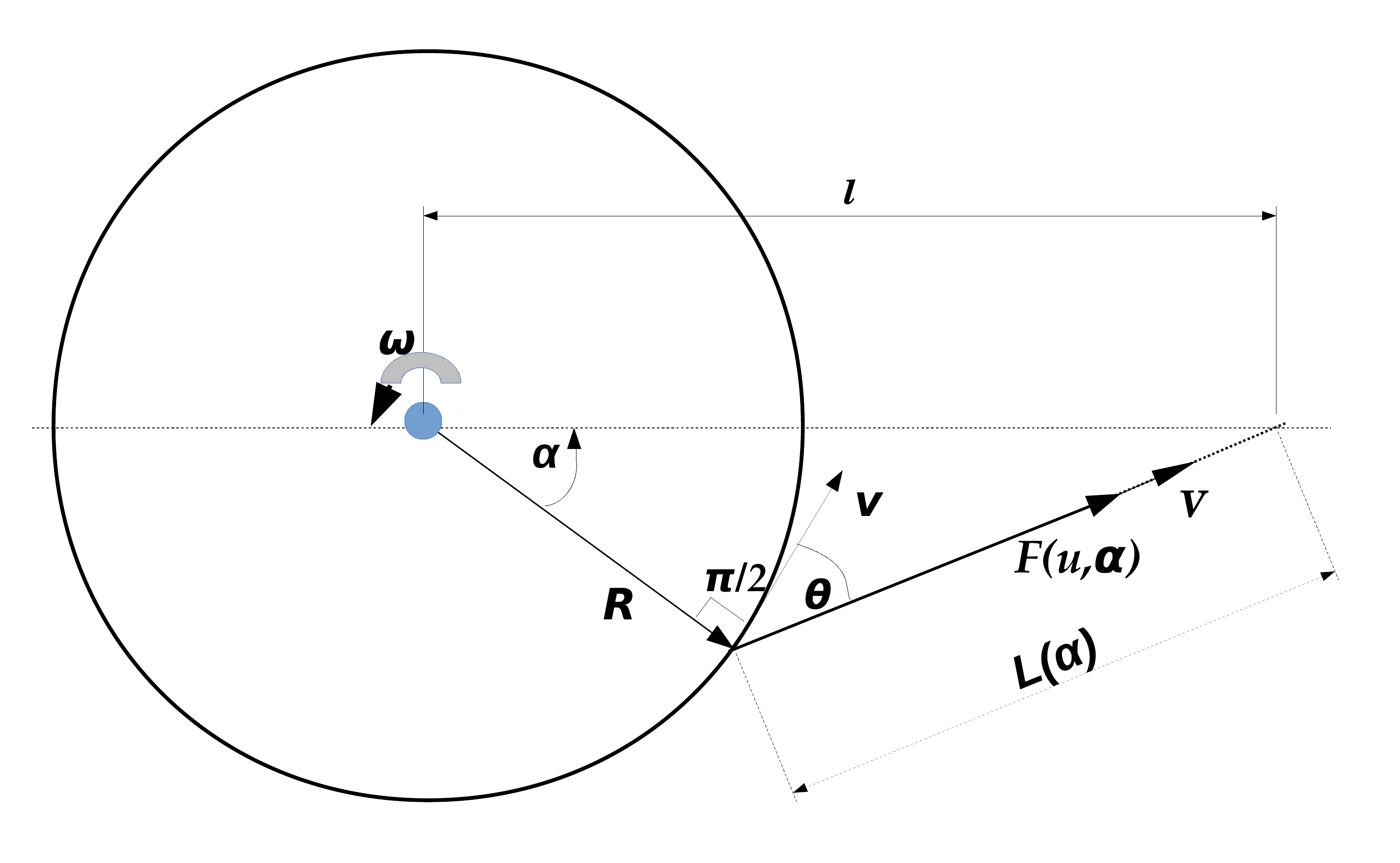}
\caption{The Planar Model of the Peeling of an Adhesive Tape off a Rotary Circular Spool.}\label{model}
\end{figure}
The radius of the spool is $R$; $l$ is the constant distance
between the center of mass of the spool and the point where the longitudinal axis of the peeled tape meets the horizontal.
The moment of inertia of the spool about its center of mass is $I$.  The angular velocity of the spool is denoted by $\omega$. 
The function $\phi(v,V)$ is the constitutive model for the adhesion depending on the constant pulling 
velocity $V$ and the peeling velocity $v$. Usually $\phi(v,V)$ is modeled
such that it increases as the tape sticks until it reaches a maximum and decreases as the tape peels off by slipping \cite{ciccotti13},\cite{de2004dynamics}. 
The parameters $E,a,l,R,I,V$ are constants with respect to time and the other variables. The time varying quantities are $\alpha,\omega,u$ and $v$.
The angular velocity $\omega$, $\dot{\omega}$,  the peel front or contact angle $\alpha$ and $\dot{\alpha}$ are positive
counter-clockwise. The model is shown in Figure \ref{model}.
\subsection{Derivation from the Lagrangian}
The model (\ref{peel}) can be formally obtained from the Lagrangian of the system. 
With reference to Figure \ref{model} and in line with the Lagrangian approach in \cite{de2004dynamics} 
the following work-energy relationships can be written in the generalized coordinate system $(\alpha, u, \int v dt)^{\mathsf T}$.
\begin{subequations}
\begin{eqnarray}
&&\mbox{Kinetic energy of the system:} \quad I \frac{(\dot{\alpha} + v/R)^2}{2} \\
&&\mbox{Potential energy input at the free end of the tape:} \nonumber \\ && \int F(u,\alpha) \dot{L}(\alpha) dt +
         \int F(u, \alpha) V dt,~\dot{L}(\alpha) = \frac{Rl \sin\alpha}{L(\alpha)} \dot{\alpha} \\ 
&&\mbox{Potential energy output at the peeling end of the tape:} \nonumber \\ &&\int F(u, \alpha) du + \int F(u, \alpha) v dt \\
&&\mbox{Energy dissipated in adhesion:} \quad \int \int \phi(v,V) dv dt  \\
&&\mbox{Lagrangian of the system:} \nonumber \\ && {\mathcal L}=I \frac{(\dot{\alpha} + v/R)^2}{2}  \nonumber \\ 
&& \quad -  \left(\int F(u,\alpha) \big(-du + \frac{Rl \sin\alpha}{L(\alpha)} d\alpha + (V-v)dt \big) - \int \int \phi(v,V) dv dt \right) ~~~~~
\end{eqnarray}
\label{lag}
\end{subequations}
from which the Euler-Lagrange equations
\begin{subequations}
\begin{eqnarray}
&& \frac{d}{dt} \frac{\partial {\mathcal L} } {\partial \dot{\alpha} } - \frac{\partial {\mathcal L} } {\partial \alpha } = 0 \\
&& \frac{d}{dt} \frac{\partial {\mathcal L} } {\partial \dot{u} } - \frac{\partial {\mathcal L} } {\partial u} = 0 \\
&& \frac{d}{dt} \frac{\partial {\mathcal L} } {\partial \dot{s} } - \frac{\partial {\mathcal L} } {\partial s} = 0,~ s := \int v dt
\end{eqnarray}
\label{lagderive}
\end{subequations}
yield 
\begin{subequations}
\begin{eqnarray}
&& I \frac{d}{dt} (\dot{\alpha} + v/R) + F(u, \alpha) \frac{Rl \sin\alpha}{L(\alpha)} + O((u/L(\alpha))^2) = 0  \nonumber \\
&& \quad \mbox{or} \quad I \dot{\omega} = -F(u, \alpha) \frac{Rl \sin\alpha}{L(\alpha)}, \label{lagtb} \\
&& \quad ~\mbox{(putting}~\omega=\dot{\alpha}+v/R,~\mbox{and ignoring higher order terms in}~u/L(\alpha).\mbox{)} \nonumber \\
&& F(u,\alpha) - \int \frac{\partial F}{\partial u} \big(V -v +\frac{Rl \sin\alpha}{L(\alpha)} \dot{\alpha}) dt = 0  \nonumber \\
&& \quad \mbox{or} \quad \dot{u} =V-v +  \frac{Rl \sin\alpha}{L(\alpha)} (1 + \frac{u}{L(\alpha)})\dot{\alpha},~\mbox{and}~ \label{lagstr} \\
&& \frac{d}{dt} I (\dot{\alpha} + v/R)/R + \phi(v,V)   -  F(u, \alpha) = 0\nonumber \\
&& \quad \mbox{or} \quad -\left(1 +  \frac{l \sin\alpha}{L(\alpha)}\right) F(u,\alpha) + \phi(v,V) = 0~~\mbox{(using (\ref{lagtb}))}~\label{lagsac}
\end{eqnarray}
\label{lageqn}
\end{subequations}
respectively. In (\ref{lag}), one can alternatively view the potential energy part of the Lagrangian 
as the elastic energy stored in the tape, $\int F(u,\alpha)(\dot{L}(\alpha) -\dot{u})dt$, to which the net work done on the system, 
$\int F(u,\alpha) (V-v) dt$, is added. 

It may be noted that when the approximations $L(\alpha) \approx l,~ u \ll L(\alpha)$ and that $\alpha$ is a small angle hold, 
(\ref{lageqn}) reduces to the same model as used in \cite{de2004dynamics}.  
The equations in (\ref{peel}) coincide with those in \cite{ciccotti13} when $\alpha$ is considered positive clockwise and $u \ll L(\alpha)$ holds. 
\subsection{Structure of the Peeling Model}
The equations of motion (\ref{av}--\ref{tb}), the material stretch rate equation (\ref{str}) and the shear-adhesion constitutive model of the peeling (\ref{sac})
constitute a DAE system in which $v$ is the algebraic variable. The variables $\omega, \alpha$ and $u$ are the differential variables.

A measure of difficulty of solving a DAE system is its 
differential index (for definition and detailed treatment see \cite{petzoldbook}). Often referred to simply as index, the differential index 
is effectively the number of times differentiations of the equations in the DAE system that should be done for obtaining a canonical first order ODE system in all the 
unknown variables of the DAE. 
Theorem 5.4.1 in \cite{petzoldbook} connects the difficulty of a numerical solution of a DAE to its differential index as the condition number of the 
Jacobian of the implicit integration method and shows it to increase exponentially with the differential index of the DAE system. 
We characterize the peeling model (\ref{peel}) 
with respect to its structure to show that the stick slip is essentially a problem of the DAE (\ref{peel}) switching to an arbitrarily large differential 
index from differential index 1 and that there is a resultant change in the time scale of evolution of the algebraic variable $v$, the peeling velocity, 
compared to the differential variables of (\ref{peel}).

We define the vectors $$
                       x := \begin{pmatrix} \alpha \cr \omega \cr  u \end{pmatrix} \in {\mathbb R}^3, ~~
f := \begin{pmatrix} \left( \omega - \frac{v}{R}\right) \cr  -\frac{R l \sin\alpha}{I L(\alpha)}  F(u, \alpha) \cr
\frac{R l \sin \alpha}{L(\alpha)}\left(1 + \frac{u}{L(\alpha)}\right)\left(\omega - \frac{v}{R}\right) + V -v \end{pmatrix}: {\mathbb R}^4 \to {\mathbb R}^3. $$
We introduce the subscript $t$ for indexing a time varying quantity, e.g., $\alpha_t := \alpha(t), ~\omega_t := \omega(t),~ u_t := u(t), 
~x_t :=x(t), ~v_t:=v(t), ~\phi_t :=\phi(t), ~F_t:=F(u_t,\alpha_t)$ and $f_t := f(x_t, v_t)$. 
The derivatives $$\frac{\partial f_t}{\partial x_t}:{\mathbb R}^4 \to {\mathbb R}^{3 \times 3},~~
\tilde{F}_t :=\frac{\partial \big((1 + \frac{l \sin\alpha}{L(\alpha)}) F(u, \alpha) \big)_t}{\partial x_t}:{\mathbb R}^2 \to {\mathbb R}^{1 \times 3}$$ are Jacobian matrices
 whereas $\frac{ \partial f_t}{\partial v_t}(\alpha_t,u_t):{\mathbb R}^2 \to {\mathbb R}^3$ is a gradient vector.
\par
The differential index of (\ref{peel}) is analyzed in the following.
\begin{lem}
 Let $J_t:= \big( -\frac{\partial \phi_t}{\partial v_t} + \tau \tilde{F}_t
 \big( {\mathcal I} - \tau \frac{\partial f_t}{\partial x_t} \big)^{-1} \frac{\partial f_t}{\partial v_t} \big),~\tau>0, \tau\to0$ exist and be non-zero
 at a time point $t$ with a $(x_t,v_t)$ satisfying (\ref{peel}).
 Then (\ref{peel}) has a unique local solution $({\mathsf x}(s),{\mathsf v}(s)),~s \in [t, t+\delta t],~ \delta t \ge 0$ as
 a function of time.\label{sol} \end{lem}
\begin{proof}
We observe that using the fundamental theorem of integral calculus, 
the Jacobian of the DAE (\ref{peel}) with respect to $x_t, v_t$ can be written as 
$\mathcal{D}_t := \begin{pmatrix}
                                                                                     {\mathcal I} - \tau \frac{\partial f_t}{\partial x_t} & -\tau \frac{\partial f_t}{\partial v_t} \cr
                                                                                   \tilde{F}_t & -\frac{\partial \phi_t}{\partial v_t}
                                                                                    \end{pmatrix},~\tau>0, \tau\to0 $ 
of which $J_t$ is the Schur complement, ${\mathcal I} \in {\mathbb R}^{3 \times 3}$ being the identity matrix. The Jacobian $\mathcal{D}_t$ 
is invertible since $J_t$ is non-zero. By the implicit function theorem, a unique solution $({\mathsf x}(s),{\mathsf v}(s))~s \in [t, t+\delta t],~ \delta t \ge 0$ 
giving $x,v$ as a function of time can be found over a neighborhood of $(t,x_t,v_t)$ if the Jacobian $\mathcal{D}_t$ is invertible.
\end{proof}
Here we remark that we primarily intend to study the local behavior of the stick-slip dynamics and the existence of a local solution suffices for the purpose.
Subject to certain conditions, one can use the Gronwal or Bihari Inequality, the Leray-Schauder Principle and Schauder Fixed Point theorem (similar to
the Peano existence for an ODE; cf. Chapter 3 of \cite{teschl})  to obtain a proof of
Peano existence and Osgood uniqueness of the global solution of (\ref{peel}) for a given consistent initial condition. However a treatment of the 
global solution of (\ref{peel}) is outside the scope of the present work.
\begin{lem}
The local differential index of (\ref{peel}) is $\lceil \log_{\tau} |J_t| \rceil + 1$ where $J_t, \tau$ are as defined in Lemma \ref{sol}.
 \label{pi}
\end{lem}
\begin{proof}
Consider the Jacobian of (\ref{peel}) with $D$ as the $d/dt$ operator. Then, by the implicit function theorem, we can rewrite $\mathcal{D}_t$ as $\begin{pmatrix}
                                                                                     {\mathcal I} - \frac{\partial (D^{-1} f)}{\partial x_t} & -\frac{\partial (D^{-1} f)}{\partial v_t} \cr
                                                                                    \tilde{F}_t & -\frac{\partial \phi_t}{\partial v_t}
                                                                                    \end{pmatrix} $ which must be invertible at a $(x_t,v_t)$ 
                                                                                    satisfying (\ref{peel}) for a unique local solution of (\ref{peel}) to exist.
Being the integration operator, $|D^{-1}|$ is $O(\tau)$,~$\tau>0, \tau\to0 $.  
It is obvious that the 2-norm of the local Jacobian $\mathcal{D}_t$ at any time point $t$ is $O(1)$.
The 2-norm of the inverse of $\mathcal{D}_t$ is the same order as that of absolute value of $J^{-1}_t$ which is the inverse of the Schur complement. 
Using the Neumann series yields \begin{eqnarray*} J_t &=& - \frac{\partial \phi_t}{\partial v_t} +  
\tilde{F}_t \frac{\partial (D^{-1} f)}{\partial v_t}
+\tilde{F}_t (\frac{\partial (D^{-1} f)}{\partial x_t}) \frac{\partial (D^{-1} f)}{\partial v_t} \\
&& + \tilde{F}_t (\frac{\partial (D^{-1} f)}{\partial x_t})^2 \frac{\partial (D^{-1} f)}{\partial v_t}
+ \cdots  \\  &=& -\frac{\partial \phi_t}{\partial v_t} +  \tau \tilde{F}_t {\partial x_t}\frac{\partial f_t}{\partial v_t} 
+ \tau^2 \tilde{F}_t \frac{\partial f_t}{\partial x_t}\frac{\partial f_t}{\partial v_t} +  \tau^3 \tilde{F}_t \big(\frac{\partial f_t}{\partial x_t}\big)^2\frac{\partial f_t}{\partial v_t}
+ \cdots.  \end{eqnarray*} The order of $|J_t|$ depends on the 2-norms of the coefficients of $\tau$ in the Neumann expansion and let this be 
$O(\tau^\nu)$ where $\nu$ is a natural number. Since $\tau>0, \tau\to0 $, $| J_t^{-1}|$ is $O(\tau^{-\nu})$ or  $O(|D|^{\nu})$. Hence the 2-norm condition
number of $\mathcal{D}_t$ is $O(|D|^{\nu})$. Thus scaling the right hand side of (\ref{peel}) with 
an operator of $O(|D|^{-\nu})$ will well-condition $\mathcal{D}_t$, which by Lemma \ref{sol} is sufficient for obtaining a unique local solution
of (\ref{peel}). This indicates $\nu $ differentiations of the equations in (\ref{peel}). Hence  $\nu + 1$, i.e., $\lceil \log_{\tau} |J_t| \rceil + 1$ differentiations
are needed to obtain a canonical ODE for $v_t$ and $\nu + 1$ is thus the differential index of (\ref{peel}). 
\end{proof}
The slip of the peeled tape occurs when the sticking resistance that has reached a maximum is overcome. The maximum adhesive force
is reached when $\frac{\partial \phi_t}{\partial v_t} = 0$ occurs in the constitutive relationship (\ref{sac}). In the following Lemma,
 we show that in the neighborhood of the maximum sticking resistance or adhesion and the subsequent slip on the yielding of the adhesive, 
(\ref{peel}) tends to have to an arbitrarily high order of singularity as characterized by its high differential index. 
This is in contrast to the regime when the adhesion is not in the neighborhood of a maximum, or,  
 $|\frac{\partial \phi_t}{\partial v_t}| \approx O(\tau^0)$ or greater implies that (\ref{peel}) has differential index one, 
since $v_t$ can be determined uniquely as a function of $(x_t,V)$ from (\ref{sac}) by the implicit function theorem.
\begin{lem} In the peeling model (\ref{peel}), let $\phi$ be a function which has at least a maximum over the range of
 values that the solution $v_t$ of (\ref{peel}) takes and $\frac{\partial \phi_t}{\partial v_t}$ monotonically and smoothly goes to zero and 
has a negative value for values of peeling velocity which are greater than $v_t$ at which $\frac{\partial \phi_t}{\partial v_t}=0$.
Then, along its solution in the neighborhood of $(t,x_t,v_t)$, (\ref{peel}) has differential index two if $\phi$ is a maximum at $v_t$. 
Further, in a neighborhood of $v_t$ at which $\phi$ is a maximum, (\ref{peel}) can have an arbitrarily large differential index.
 \label{hi}
\end{lem}
\begin{proof}
 When $\frac{\partial \phi_t}{\partial v_t}= 0$ at the maximum, the peeling velocity solution $v(s),~s \in (t,t+\delta t],~\delta t > 0$ 
cannot be determined from (\ref{sac}), i.e., one differentiation is not sufficient to
 obtain a notional first order ODE in $v$. Hence the differential index of (\ref{peel}) is higher than one if $\frac{\partial \phi_t}{\partial v_t}= 0$  . \par
By the implicit function theorem and from (\ref{av}--\ref{str}), we obtain 
\begin{eqnarray} \frac{\partial x_t}{\partial v_t} = \tau 
 \big( {\mathcal I} - \tau \frac{\partial f_t}{\partial x_t} \big)^{-1} \frac{\partial f_t}{\partial v_t} \label{xinv}
 \end{eqnarray} where $\tau \to 0, ~\tau > 0$. 
 Then, by Lemma \ref{sol}, $v$ can be determined explicitly
 in terms of $x$ and $t$ in the neighborhood of $(t,x_t,v_t)$ if 
 $J_t =\big( -\frac{\partial \phi_t}{\partial v_t} 
 + \frac{\partial \big( (1 +\frac{l \sin\alpha}{L(\alpha)}) F(u,\alpha) \big)_t}{\partial x_t}\frac{\partial x_t}{\partial v_t} \big)$ 
 is non-zero. 
If $\frac{\partial \phi_t}{\partial v_t} = 0$ at $v_t$,  and a $1 \gg \tau > 0$ is chosen, then 
\begin{eqnarray*}
J_t\bigg|_{\frac{\partial \phi_t}{\partial v_t} = 0} &=& \big( -\frac{\partial \phi_t}{\partial v_t} + \tau \tilde{F}_t
 \big( {\mathcal I} - \tau \frac{\partial f_t}{\partial x_t} \big)^{-1} \frac{\partial f_t}{\partial v_t} \big) = \tau \tilde{F}_t
 \big( {\mathcal I} - \tau \frac{\partial f_t}{\partial x_t} \big)^{-1} \frac{\partial f_t}{\partial v_t} \\
&=& \tau  \tilde{F}_t \frac{\partial f_t}{\partial v_t} + \tau^2 \tilde{F}_t \frac{\partial f_t}{\partial x_t}  \frac{\partial f_t}{\partial v_t} 
+ \tau^3 \tilde{F}_t \big(\frac{\partial f_t}{\partial x_t}\big)^2 \frac{\partial f_t}{\partial v_t} +  \cdots \\
&=& \tau \bigg( \frac{k l u \sin (\alpha ) \left(l^2+R^2\right)}{L(\alpha )^3 (u-L(\alpha ))^2}+\frac{k l u \cos (\alpha ) \left(l^2+R^2\right)}{R L(\alpha )^3 (u-L(\alpha ))}-\frac{k l^2 R u \sin (2 \alpha )}{L(\alpha )^3 (u-L(\alpha
   ))^2} \\ && -\frac{k l^2 u^2 (\cos (2 \alpha )+3)}{2 L(\alpha )^3 (u-L(\alpha ))^2}+\frac{2 k l^2 u}{L(\alpha )^2 (u-L(\alpha ))^2} \\ 
   &&-\frac{k (l \sin (\alpha )+L(\alpha )) \left(l L(\alpha ) \sin (\alpha )+l u \sin (\alpha
   )+L(\alpha )^2\right)}{L(\alpha )^2 (u-L(\alpha ))^2} \bigg) + O(\tau^2)
         \end{eqnarray*}
which is $O(\frac{k}{l}\tau)$ and hence, by Lemma \ref{pi}, (\ref{peel}) has differential index two.
\par
Let $\delta t$ and $\delta v_t$ be small positive numbers.
If $\phi(v,V)$is a maximum at $v=v_t$, then in a neighborhood $v_t \pm \delta v_t$, $\big| \frac{\partial \phi_t}{\partial v_t} \big| \le \epsilon,~0<\epsilon \ll 1$ holds
since $ \frac{\partial \phi_t}{\partial v_t}$ decreases monotonically from positive value to zero as $v$ decreases to $v_t$. 
For the values of $v \in (v_t, v_t+\delta v_t)$,  $\frac{\partial \phi_t}{\partial v_t} $ is 
a small negative number by the property of the function $\phi(v,V)$.  Then, $\phi(v_s,V), ~v_s \in \{[v_t-\delta v_t, v_t), (v_t, v_t+\delta v_t)\}$ 
at some time point $s \in  \{[t-\delta t, t), (t, t+\delta t)\}$ is such that
$$\left|- \frac{\partial \phi_s}{\partial v_s} + 
\sum_{p=1}^m \tilde{F}_s  \big(\frac{\partial f_s}{\partial x_s}\big)^{p-1}\frac{\partial f_s}{\partial v_s} \tau^p \right| \approx O(\tau^{m+1})~ \mbox{or smaller} ,~m \ge 1,$$
and $|J_t| \approx O(\tau^{m+1})$. The differential index of (\ref{peel}) is consequently $m+2$ by Lemma \ref{pi}. 
Since $\frac{\partial \phi_s}{\partial v_s}$ goes to zero smoothly and monotonically and
since $\tau \to 0, ~\tau > 0$, $\frac{\partial \phi_s}{\partial v_s}$ can cancel the first $m$ arbitrary large number of terms of the Neumann series expansion.
\end{proof}
The local Jacobian $\mathcal{D}_t$ becomes rank deficient as the differential index tends to
be arbitrarily high and $J_t \to 0$. From Lemma \ref{sol} and \ref{pi}, then (\ref{peel}) will no longer have
a unique local solution over $(t, t+\delta t]$ for some suitable $\delta t > 0$.
\par From the above Lemmas and the rank deficiency of the Jacobian $\mathcal{D}_t$ at the stick-slip, we conclude the following.
\begin{thm}
 The peeling model (\ref{peel}) is a DAE that has variable differential index which is at least one and can be arbitrarily high making the  
local solutions
 of (\ref{peel}) non-unique at time points leading up to and following the maximum adhesion. \label{ithm}
\end{thm}
Physically, the high differential index near the maximum adhesion
and the slip affects the process modeled by (\ref{peel}) over short time sub-intervals over which 
$|J_t| \approx O(\tau^\nu)$, $\nu+1$ being the local differential index of (\ref{peel}). Thus, the two time scales emerge with respect to
the stick-slip dynamics of the peeling of an adhesive tape: one during when differential index is one and the other when
the differential index rapidly increases to an arbitrarily high value. The second time scale is of 
interest with respect to a study of nonlinear bifurcation, the behavior of the peeling velocity $v$ at slip and the 
homogenization of $v$ over the same time scale. 
% and 
% also any experimental realization of the peeling of an adhesive tape which is limited by the finite frequency at which measurements can
% be taken.   
\section{Time Scale of the Stick-Slip Dynamics}
In this section we investigate the time scale of the stick-slip process.
\begin{lem}
Let $1 \gg \tau > 0$.  Then \begin{eqnarray}
 J_t ~\delta v_t & = &- \frac{\partial \big((1 +\frac{l \sin\alpha}{L(\alpha)}) F(u,\alpha) \big)_t}{\partial x_t} \delta x_t           \label{dvdx}                 
                            \end{eqnarray}
where $J_t$ is as defined in Lemma \ref{sol}. \label{dvdxlem}
\end{lem}
\begin{proof}
 Taking total differentials on (\ref{sac}) we get 
 \begin{eqnarray} && \big( -\frac{\partial \phi_t}{\partial v_t} + 
\frac{\partial \big( (1 + \frac{l \sin\alpha}{L(\alpha)}) F(u,\alpha) \big)_t}{\partial x_t}\frac{\partial x_t}{\partial v_t} \big) \delta v_t = 
-\frac{\partial \big((1 + \frac{l \sin\alpha}{L(\alpha)}) F(u,\alpha) \big)_t}{\partial x_t} \delta x_t  \nonumber \end{eqnarray}
where $\frac{\partial x_t}{\partial v_t} = \tau  \big( {\mathcal I} - \tau \frac{\partial f_t}{\partial x_t} \big)^{-1} \frac{\partial f_t}{\partial v_t}$ from Lemma \ref{hi}.
\end{proof}
\begin{lem}
 Let $ |J_t| \ne 0$ be $O(\tau^{\nu}), ~1 \le \nu < \infty, ~ 0 < \tau \ll 1$, and $\eta_0\in [0,1]$. 
Then,
\begin{eqnarray} \sqrt{ \eta_0^2  |\alpha(s)-\alpha(t)|^2 +(1-\eta_0^2) |u(s)-u(t)|^2} \le K \tau^{\nu}  |v(s) - v(t)|, \label{tseq} \end{eqnarray} 
where $s = \arg \sup_{s \in [t, t+\delta t]} \sqrt{ |\alpha(s)-\alpha(t)|^2 +  |u(s)-u(t)|^2}$ with $\delta t \to 0, \delta t > 0$ 
at any time point $t \in [0, \infty)$.
\label{ts}
\end{lem}
\begin{proof}
We observe that $\|\tilde{F} \|_2 $ is of $O((k/l))$. From (\ref{dvdx}) in Lemma \ref{dvdxlem}, 
we obtain by pre-multiplying both sides with $(\tilde{F}^{\mathsf T}\tilde{F} )^{+}\tilde{F}^{\mathsf T}$,
\begin{eqnarray} (\tilde{F}^{\mathsf T}\tilde{F} )^{+}\tilde{F}^{\mathsf T} J_t \delta v_t + \Pi \delta x_t = 0  \label{lsqx} \end{eqnarray}
where $^{+}$ is the pseudo-inverse and $\Pi := (\tilde{F}^{\mathsf T}\tilde{F} )^{+}\tilde{F}^{\mathsf T}\tilde{F}=qq^{\mathsf T}$ 
is a rank one matrix, $q$ being the unit vector:
\begin{eqnarray*}
q:=\tilde{F}_t^{\mathsf T}/\|\tilde{F}_t \|_2 \approx
\left(
\begin{array}{ccc}
 \frac{ l^2 \left(\left(l^2+R^2\right) \cos (\alpha) -R \left(2 l+L(\alpha) \sin (\alpha )\right)\right) }
{L^2(\alpha) \left(l \sin (\alpha )+ L(\alpha)\right)}\frac{u}{l}+O\left(\left(\frac{u}{l}\right)^2\right) \cr 0 \cr
1-O\left(\left(\frac{u}{l}\right)^2\right)  \\
\end{array}
\right).
\end{eqnarray*}
Using the structure of $q$, and by taking norms we obtain from the right hand side of (\ref{lsqx})
\begin{eqnarray*}
 \|\Pi (x(s) - x(t)) \|_2  & =& \sqrt{ \eta_0^2 |\alpha(s)-\alpha(t)|^2 + (1-\eta_0^2) |u(s)-u(t)|^2}.
\end{eqnarray*}
in which $\eta_0\in [0,1]$ is a constant independent of $\tau$.
Then, taking norms on the both sides of (\ref{lsqx}) and by applying the Cauchy-Schwarz inequality, we get
\begin{eqnarray*}
 \sqrt{ \eta_0^2 |\alpha(s)-\alpha(t)|^2 + (1-\eta_0^2) |u(s)-u(t)|^2} 
  &=& \big\| -  (\tilde{F}^{\mathsf T}\tilde{F} )^{+}\tilde{F}^{\mathsf T} J_t \delta v_t  \big\|_2 \\
    & \le &  K_1 \tau^{\nu} K_2 \frac{k}{l} |v(s) - v(t)| \nonumber
\end{eqnarray*}
where $K_1, ~K_2$ are constants independent of $\tau$.
\end{proof} 
From the above Lemmas \ref{dvdxlem} and \ref{ts} together with Lemma \ref{pi}, and since $0 < \tau \ll 1$, 
we conclude the following result.
\begin{thm}
 In (\ref{peel}) the  peeling velocity $v_t$ changes exponentially faster than the time scale of change of either the peel front angle or 
the tape's tensile displacement or a linear combination of both, the exponent of the time scale being one less than the local 
differential index of (\ref{peel}). \label{expv}
\end{thm}
As a corollary to Theorem \ref{expv} and from Lemma \ref{hi} it is obvious that in the neighborhood of the time points at which
$\big| \frac{\partial \phi_t}{\partial v_t} \big| \le \eta \ll 1$, i.e., at the stick-slips, 
the magnitude of the change in the peeling velocities can be arbitrarily high. Thus the peeling velocity
undergoes a (relatively) stiff change when the maximum adhesion is approached or just following the slip. 
\par
Consequently, we have two distinct regimes in the dynamics of peeling of an adhesive tape: 
\begin{itemize}
 \item a regime during which the variables $\alpha,~u$ change with respect to time in the same scale as the peeling velocity $v$. 
This happens when the local differential index of (\ref{peel})
is one, i.e., $\big| -\frac{\partial \phi_t}{\partial v_t} + \tau \tilde{F}_t
 \big( {\mathcal I} - \tau \frac{\partial f_t}{\partial x_t} \big)^{-1} \frac{\partial f_t}{\partial v_t} \big| \approx O(\tau^0)~\mbox{or greater},~\tau>0, \tau\to0  $. 
We call this the {\it slow scale}.
\item another regime when $\big|-\frac{\partial \phi_t}{\partial v_t} + \tau \tilde{F}_t
 \big( {\mathcal I} - \tau \frac{\partial f_t}{\partial x_t} \big)^{-1} \frac{\partial f_t}{\partial v_t} \big| \approx O(\tau^m),~\tau>0, \tau\to0,~m>0$,
 i.e., when the local differential index of (\ref{peel}) 
is greater than one. This is the stick or adhesion regime followed by the slip, during which the peeling velocity changes 
exponentially faster compared to the change in $\alpha$ or $u$ or in both . We shall refer to this as the {\it fast scale}.
\end{itemize}
We remark that as a straightforward consequence of Lemma \ref{ts} a reformulation of (\ref{sac}) as a an ODE 
in which a small positive quantity is multiplied with the time rate of change of the peeling velocity may not always 
satisfy the constitutive relationship (\ref{sac}), especially, in the fast scale of the stick-slip regime when the rate of change of the peeling velocity 
is as large as the inverse of the small multiplier. 

In this context, we consider the ODE formulation in \cite{de2005missing}, which is arrived at by 
considering an additional kinetic energy term in the Lagrangian due to the stretch rate of the extremely small mass of the tape and 
not by simply introducing a singular ODE. The constitutive relationship (\ref{sac}) then is an ODE and not an algebraic equation.
\begin{eqnarray}
 m \ddot{u} & = & \bigg(1+ \frac{l \sin\alpha}{L(\alpha)}\bigg) F(u,\alpha) - \phi (v,V),   \nonumber 
\end{eqnarray}
where $0 <m \ll 1$ is a small mass, involves $m \dot{v}_t$ on the left hand side. If $m$ is such that 
$|\dot{v}_t| $ grows as $O(1/m)$ or faster at the stick-slip points at which the Schur complement tends to go to zero, 
then, $\bigg(1 + \frac{l \sin\alpha}{L(\alpha)}\bigg) F(u,\alpha) - \phi (v,V)$ does not necessarily go to zero in the high differential index stick-slip regime. 
Also, $|\alpha|$ and $|\dot{\alpha}|$ do not grow as fast as $v_t$ at the the stick-slip points with high local differential index. In such a case 
the role of the small mass becomes that of an inertial regularization parameter relative to the nearly singular ODE model of the peeling dynamics. 
If an $m \ll 1/|\dot{v}_t|$ is chosen, this formulation has the same characteristics as a high differential index constraint. 
In the slower scale without the stick-slip (in which the local differential index is one) $v_t$ does not change exponentially faster in time 
and the above formulation satisfies (\ref{sac}) in the limit as $m \to 0$. However, in the present work we have assumed
that the stretched tape is mathematically massless, i.e., any kinetic energy of the stretching tape is negligible. 
This is consistent with the standard experimental set-up 
described in works such \cite{ciccotti13}. Thus our approach obtains the DAE (\ref{peel}) and concerns about its behavior and regularization 
in the high differential index regime.  
\section{The Stick-Slip Dynamics}
In this section we see if the dynamics in the fast scale during the intermittent stick-slip affect nonlinear bifurcation in this regime. 
As $|J_t|\to 0$, we show that the model (\ref{peel}) is driven by the kinetics and kinematics (\ref{av}--\ref{str}) and not the adhesion-shear constitutive relationship
(\ref{sac}). As the tape sticks almost the maximum, due to the rotational inertia of the spool peel front angle $\alpha$ changes and due to
the constant pulling velocity $V$ the tensile displacement $u$ changes, perturbing the dynamics (\ref{av}--\ref{str}).
\begin{lem}
 Let $(x_t,v_t)$ satisfy (\ref{peel}) at $t$ such that $|J_t| \approx \tau^\nu,~ \nu \gg 1$ and $u \ne 0$. 
A perturbation $|\delta \alpha_t| \tau^{-\nu} \ne 0$ of the  peel front angle $\alpha_t$ or $|\delta u_t| \tau^{-\nu} \ne 0$
of the tensile displacement $u_t$ or a combination of both regularizes the Jacobian $\mathcal{D}_t$.
\label{reglem}
 \end{lem}
\begin{proof} 
By (\ref{tseq}) in Lemma \ref{ts}, a sufficient perturbation  $|\delta \alpha_t| \tau^{-\nu} \ne 0$ of the  peel front angle $\alpha_t$ or $|\delta u_t| \tau^{-\nu} \ne 0$
of the tensile displacement $u_t$ or a combination of both can affect a non-zero change $\delta v_t$, since
$\eta_0 > 0$ in (\ref{tseq}) for $u \ne 0$. When
$\frac{\partial \phi}{\partial v}(v_t+\delta v_t, V)$ is $O(\tau^0)$ or more, the Schur complement $J_t$ 
of the Jacobian ${\mathcal D}_t$ of (\ref{peel}) 
becomes non-zero
and ${\mathcal D}_t$ becomes invertible due to this regularization affected by the perturbation. \end{proof} 
Regularization of the Jacobian ${\mathcal D}_t$ of (\ref{peel}) implies that the DAE (\ref{peel}) can 
have a solution over some time interval containing the time point $t$ in the stick-slip regime. 
However this solution is not unique since the condition number of ${\mathcal D}_t$ now depends on the perturbation
$|\delta \alpha_t|$ from the rotational inertia of the spool or $|\delta u_t|$ of the tensile diplacement 
of the tape due to the constant $V$ or both. Lemma \ref{reglem} 
shows how the dynamics at stick-slip may continue drawing from the perturbations from rotational inertia
 of the spool or from the constant pulling velocity or both, rather than the relaxation of the constitutive relationship (\ref{sac}). 
 \subsection{Local Nonlinear Bifurcation}
 Let $(x_t,v_t)$ at a time point $t$ in the stick-slip regime satisfy (\ref{peel}) such that the local differential index of (\ref{peel}) is 
 significantly more than unity and possibly arbitrarily large. 
Let the DAE (\ref{peel}) have a solution $({\mathsf x}(s),{\mathsf v}(s)),~s \in {\mathcal T} 
\subset {\mathbb R}$, ${\mathcal T}$ being a time interval containing $t$. 
 Then, following \cite{deimling}(cf. Chapter 10, Definition 28.1), we define $(x_t,v_t)$ as a {\it local nonlinear bifurcation} point of (\ref{peel}) 
iff $({\mathsf x}(s),{\mathsf v}(s), s \in{\mathcal T} ) = \lim_{n \to \infty} ({\mathsf x}^{(n)}(s), {\mathsf v}^{(n)}(s),  s \in {\mathcal T}^{(n)})$ 
hold with $({\mathsf x}^{(n)}(s), {\mathsf v}^{(n)}(s),  s \in{\mathcal T}^{(n)})$ being a solution of (\ref{peel}) over a time interval 
${\mathcal T}^{(n)}$ containing $t$ for each $n=1, 2, \cdots$ such that 
$({\mathsf x}^{(n)}(s), {\mathsf v}^{(n)}(s),  s \in{\mathcal T}^{(n)}) \ne ({\mathsf x}(s),{\mathsf v}(s), s \in{\mathcal T} )$ for all $n$.
 \begin{thm}
Let $(x_t,v_t)$ together with $|J_t| \to 0$ satisfy (\ref{peel}) in the fast scale of the stick slip dynamics. Then 
(\ref{peel}) has a local nonlinear bifurcation point at $(x_t,v_t)$ in the stick-slip regime of the peeling of an adhesive tape. \label{bif}
 \end{thm}
 \begin{proof}
  Let $\alpha_t$ or $u_t$ or both be perturbed as in the condition of 
  Lemma \ref{reglem} so that $\delta x_t^{(i)}, i=1, \cdots, n$ is non-zero in magnitude
  and that  $0 < \|\delta x_t^{(i+1)}\| <  \|\delta x_t^{(i)}\|$ with
  $\|\delta x_t^{(n)}\| \to 0$ ( $\|.\|$ being a suitable vector p-norm). Then, by Lemma \ref{reglem}, $|J_t| \to 0$, 
 is perturbed successively as $\big| J_t + \delta J_t^{(i)}\big|> 0, ~i =1, \cdots, n$. 
Due to the perturbation of the Schur complement $J_t$, we obtain a sequence of invertible Jacobian matrices 
${\mathcal D}_t^{(i)},~i =1, \cdots, n$ of (\ref{peel}). 
By the implicit function theorem, each perturbation leads to the existence of a unique solution $({\mathsf x}^{(i)}(s), {\mathsf v}^{(i)}(s))$
of (\ref{peel}) over a time interval ${\mathcal T}^{(i)}$ containing the time point $t$ at which $x_t, v_t$ satisfies (\ref{peel}). In the limit as $n \to \infty$, 
$\delta x_t^{(n)} \to 0$ so that $({\mathsf x}^{(n)}(s), {\mathsf v}^{(n)}(s))$ tends to  $({\mathsf x}(s),{\mathsf v}(s))$ and
${\mathcal T}^{(n)} \to {\mathcal T}$ with $| {\mathcal T} |$ tending to be arbitrarily small 
(due to the Schur complement approaching zero).
Then by the definition stated above, $(x_t,v_t)$ is a local nonlinear bifurcation point of (\ref{peel}).
 \end{proof}
While the perturbation due to rotational inertia and constant pulling velocity advances
 the dynamics, by Lemma \ref{ts} the time scale in which the peeling velocity changes is exponentially faster than that of $\alpha_t$ and/or $u_t$.
This indicates that the peeling velocity becomes highly sensitive to small perturbations at the local nonlinear bifurcation point. 
Thus there is a near jump in the peeling velocity with the shear force peeling the tape remaining almost unchanged.
The shear force is a function of the differential variables $\alpha$ and $u$ that change exponentially slowly compared to the peeling velocity. 
We conjecture that the exponentially fast jump like change in the peeling velocity in the stick-slip regime contributes to
the experimentally observed intense release of energy in acoustic or triboluminescence \cite{camara} form due the instantaneous breaking of the molecular bonds 
in the process of shearing of the adhesive over an almost negligible time interval.
 \section{Numerical Simulation: Homogenization of the Peeling Velocity}
It is obvious that capturing the arbitrarily large changes in the peeling velocity at the local nonlinear bifurcation points of (\ref{peel}) 
is difficult because of the possibly arbitrarily high local differential index of (\ref{peel}). 
Most DAE solvers cannot cope with differential index greater than $3$ due to Theorem 5.4.1 in \cite{petzoldbook}. 
A computational method involving finding repeatedly the bifurcation points in time in order to 
deal with the singular points by stopping and restarting the DAE integration algorithm with 
regularization is expensive especially for those $V$'s at which the stick-slip regime dominates.
At a pull velocity, $V$, for which the peeling dynamics is mostly in the fast scale of stick slip regime, 
a direct numerical simulation of the DAE (\ref{peel}) may be thus difficult and likely inefficient compared to the computational effort. 
This calls for a consistent reformulation of (\ref{peel}) so that the numerical integration of 
the reformulated algebraic constraint of (\ref{peel}) would correctly average out (weak convergence) the actual peeling velocity response
at the bifurcation points. 
\par
We assume that the pulling velocity $V$ is such that the peeling dynamics has intermittent 
stick-slip, i.e., both the fast and slow scales. We choose a $0 < \tau \ll 1 $ such that
 $J_s$ is invertible over a small time interval of length $\delta t$ at every $s \in [t, t+\delta t], ~ \delta t >0$. 
Then $\tau$ acts as the characteristic homogenization time scale in (\ref{dvdx}) and we re-write (\ref{dvdx}) in the following integral form over 
the time interval $[t, t+\delta t]$, 
\begin{subequations}
\begin{eqnarray}
 v_r & \approx &  v_t + \int_t^r  -J_s^{-1} \frac{\partial \big((1 + \frac{l \sin\alpha}{L(\alpha)}) F(u,\alpha) \big)_s}{\partial x_s} dx_s    \nonumber \\
 & \approx & \int_t^{r} -\bigg( -\frac{\partial \phi_s}{\partial v_s} 
      + \tau \tilde{F}_s \big( {\mathcal I} - \tau \frac{\partial f_s}{\partial x_s} \big)^{-1}\frac{\partial f_s}{\partial v_s}\bigg)^{-1} \tilde{F}_s f_s ds + v_t  \label{hv} \\
&\approx & \int_t^r -\bigg( -\frac{\partial \phi_s}{\partial v_s} 
+ \tau  \tilde{F}_s \frac{\partial f_s}{\partial v_s} + \tau^2 \tilde{F}_s\frac{\partial f_s}{\partial x_s}\frac{\partial f_s}{\partial v_s} +\tau^3  \tilde{F}_s\big(\frac{\partial f_s}{\partial x_s}\big)^2 \frac{\partial f_s}{\partial v_s} 
+ \cdots\bigg)^{-1} \tilde{F}_s f_s ds + v_t ~~~~~~~~~\label{hvex}
\end{eqnarray}
\label{hvv}
\end{subequations}
where $r \in [t, t+\delta t]$ and $ -J_s^{-1} \frac{\partial \big((1 + \frac{l \sin\alpha}{L(\alpha)}) F(u,\alpha) \big)_s}{\partial x_s}$ acts 
as a projection of the slower dynamics of differential variables $x$ onto the faster stick-slip dynamics of the peeling velocities $v$. 
The reformulation of (\ref{peel}) involving the homogenized peeling velocity is then obtained by appending (\ref{hv}) to (\ref{av}--\ref{str}).
By Lemma \ref{ts}, the integral equation (\ref{hv}) can be seen as a scaling of the time of relaxation of the peel front angle and tensile displacement
of the tape by $\tau^{\nu}$ to match that of the peeling velocity. If there is 
no stick-slip at any $s \in [t, t+\delta t]$, then $J_s$ will be invertible even with $\tau \to 0$ by the Lemmas \ref{pi} and \ref{hi}. On the other hand, if there is a 
bifurcation point in $[t, t+\delta t]$, then (\ref{peel}) will have a very high differential index at the bifurcation points by the Lemmas \ref{pi} and \ref{hi}
and consequently, (\ref{hv}) can be written only with a  $\tau$ significantly greater than zero if $J_s$ must be invertible at each $s \in [t, t+\delta t]$. 
Thus the minimum $\tau$ at which the locally high differential index DAE (\ref{peel}) can be integrated by a specific,
at least $A$-stable (cf. \cite{hairer2}) implicit numerical method to a prescribed accuracy will indicate the smallest time scale to which the solution of 
(\ref{peel}) can be numerically resolved by that particular method. 
\subsection{Multiple-scale Expansion}
The multiple time scale expansion of the homogenized peeling velocity $v$ is provided by (\ref{hvex}).
\begin{thm}
Let $0 < \tau \ll 1$ and $s \in [t, t+\delta t]$, $\delta t > 0$ being a suitable small time interval, 
be such that $$\tilde{J}_s :=- \frac{\partial \phi_s}{\partial v_s} + 
\sum_{p=1}^m \tilde{F}_s  \big(\frac{\partial f_s}{\partial x_s}\big)^{p-1}\frac{\partial f_s}{\partial v_s} \tau^p,~ \left|\tilde{J_s} \right| \approx O(\tau^{m}),~m \ge 0,$$
and that $J_s$ is invertibe. Also, let $J_s = \tilde{J}_s + \tilde{J}^{(r)}_s$. Then, for $r \in [t, t+\delta t]$,
\begin{eqnarray}
v_r & \approx & \int_t^r -\tilde{J}_s^{-1}(1 - \tilde{J}_s^{-1}(-\tilde{J}^{(r)}_s))^{-1}\tilde{F}_s f_s ds + v_t\nonumber \\
& \approx & \int_t^r -\tilde{J}_s^{-1}(1 - \tilde{J}_s^{-1}\tilde{J}^{(r)}_s + \tilde{J}_s^{-2}(\tilde{J}^{(r)}_s)^2 - \cdots) \tilde{F}_s f_s ds + v_t \nonumber \\
& \approx & v_t + \int_t^r \tilde{F}_s f_s (\Upsilon_1 \tau^{-m} + \Upsilon_2 \tau^{-m+1} + \cdots)_s ds,~\Upsilon_{i} \in {\mathbb R}~\mbox{independent of}~\tau,~i=1,2,\cdots~~~~~~
\label{msv}
% v_r & \approx &  \int_t^r -\bigg(\frac{1}{\tilde{J}_s} -
%   \frac{\tilde{F}_s\big(\frac{\partial f_s}{\partial x_s}\big)^{m}\frac{\partial f_s}{\partial v_s}}{{\tilde{J}^2}_s} \tau^{m+1}+ \nonumber \\
%  && \bigg( \frac{\big( \tilde{F}_s\big(\frac{\partial f_s}{\partial x_s}\big)^{m}\frac{\partial f_s}{\partial v_s}\big)^2}{{\tilde{J}^3}_s} - 
%  \frac{\tilde{F}_s\big(\frac{\partial f_s}{\partial x_s}\big)^{m+1}\frac{\partial f_s}{\partial v_s}}{{\tilde{J}^2}_s} \bigg) \tau^{m+2} + \cdots\bigg) \tilde{F}_s f_s ds \nonumber \\
%  & \approx & \int_t^r -\big(\Upsilon_0 \tau^{-m} + \Upsilon_1 \tau^{-m+1} + \Upsilon_2 \tau^{-m+2} + \cdots \big) \tilde{F}_s f_s ds
\end{eqnarray} 
where $\Upsilon_{.}$ are the coefficients in the expansion (\ref{msv}).
\label{msvexp}
\end{thm}
\begin{proof}
Follows immediately from (\ref{hvv}). 
\end{proof}
Each function $\tilde{F}_s f_s \Upsilon_i,~i=1,2,\cdots$ is the generalized time derivative of the component of the peeling velocity in the $(m-i+1)$th scale
with $\tau^{-m+(i-1)}$ as the test function. It is in this sense that the multiple time scale expansion indicates homogenization of the peeling 
velocity.

As $\tau \to 0$, the larger eigenvalues including the higher frequencies are captured in the multiple scale expansion (\ref{msv}). As $\tau \to 1$, $v$ is averaged to
the same time scale as the differential variables and the oscillations are damped out since $\tilde{F}_t \frac{\partial x_t}{\partial v_t}$ provides significant 
damping in (\ref{hvv}) to the effect of the shear force peeling the tape.

We remark here that (\ref{hvv}) is obtained by taking total differentials on (\ref{sac}) using the implicit function theorem; and is not an artifact of
derivation from the system Lagrangian, i.e., it does not alter the underlying physics of the system but simply exploits the mathematical structure
of the DAE (\ref{peel}) to introduce the characteristic time scale and the resultant homogenized ODE (\ref{hv}).
\subsection{Example of Homogenization}
The homogenization of $v$ can be elucidated by considering the following case. 
Let $(x_t,v_t)$ satisfying (\ref{peel}) at $t$ be such that $\big| \frac{\partial \phi_t}{\partial v_t}\big| < \eta,~ 0 \le \eta \ll 1$, and 
for a $0 < \tau \ll 1$, let $\big| -\frac{\partial \phi_t}{\partial v_t} + \tau \tilde{F}_t\frac{\partial f_t}{\partial v_t}  + 
\tau^2 \tilde{F}_t\frac{\partial f_t}{\partial x_t}\frac{\partial f_t}{\partial v_t} \big|$ be $O(\tau^4)$ and
 $\big|\tilde{F}_t \left(\frac{\partial f_t}{\partial x_t}\right)^2 \frac{\partial f_t}{\partial v_t}\big|$ be $O(\tau^0)$. 
Then the local differential index is $4$ and (\ref{peel}) is already in the fast scale of the stick slip dynamics. 
Further, suppose $\tau^4 \to 0$ but not $\tau^3$. If $[t, t+\tilde{t}]$ be such a time interval that the above conditions hold at any $s \in [t, t+\tilde{t}]$, then,
in the said time interval, using the re-formulation (\ref{hv}) of (\ref{sac}) would lead to the following homogenization as indicated by the multiscale 
expansion (\ref{hvex}).
\begin{eqnarray}
 v_r & \approx & \int_t^r  \bigg( -\tau^3 \tilde{F}_s \big(\frac{\partial f_s}{\partial x_s}\big)^2 \frac{\partial f_s}{\partial v_s}
\bigg)^{-1}f_s ds + v_t \nonumber \\
& \approx & \frac{r-t}{\tau^3} \frac{\int_t^r \bigg(-\tilde{F}_s \big(\frac{\partial f_s}{\partial x_s}\big)^2 \frac{\partial f_s}{\partial v_s}
\bigg)^{-1} f_s ds}{(r-t)} + v_t,~r \in (t, t+\tilde{t}] . \nonumber
\end{eqnarray}
It may be noted that the homogenization time scaling is exponential in $\tau$, i.e., exponent is $3$ in this specific example.
The time scale for changes in the peeling velocity are thus squeezed as $1/\tau^3$, scaling the average peeling velocity over the 
slower time scale of the differential variables to a faster one.
\subsection{Relationship with Bifurcation} 
Let $0 < \tau^{(i+1)} < \tau^{(i)} \ll 1,~i=1, \cdots, n$ be successive homogenization time scales such that $J_s$ is invertible at any $s \in [t, t+\delta t]$.
By Lemma \ref{sol}, this makes it possible to obtain numerical solutions $(\bar{\mathsf{x}}^{(i)}(s), \bar{\mathsf{v}}^{(i)}(s)),~i=1, \cdots, n$. 
Obviously $\tau^{(n)} \ge \tau_*$ where $\tau_*$ is the least positive real number for which $J_s$ is invertible at any $s \in [t, t+\delta t]$. 
As $n \to \infty$, we can make $\tau^{(n)} \to \tau_*$.
Corresponding to this, we obtain $\lim_{n \to \infty} (\bar{\mathsf{x}}^{(n)}(s), \bar{\mathsf{v}}^{(n)}(s)) \to  (\bar{\mathsf{x}}(s), \bar{\mathsf{v}}(s))$ where
$ (\bar{\mathsf{x}}(s), \bar{\mathsf{v}}(s))$ is the homogenized solution of (\ref{peel}) using (\ref{av}--\ref{str}) and (\ref{hv}) for $s\in [t, t+\delta t]$. 
Thus the reformulation  (\ref{hv}) along with (\ref{av}--\ref{str}) also captures the local nonlinear bifurcations, if any, at a suitable $\tau_*$. 
The homogenization time scale $\tau$ effectively perturbs the differential variables
in the Jacobian matrix $\frac{\partial f_s}{\partial x_s}$ and the gradient vector $\frac{\partial f_s}{\partial v_s}$ to affect a 
non-zero Schur complement,
which in turn, produces a regularized Jacobian matrix ${\mathcal D}_t$. Then, by Lemma \ref{reglem} and Theorem \ref{bif}, one can conclude that 
the homogenization is not inconsistent with the local nonlinear bifurcation.
\subsection{Numerical Integration}
The very high differential index of the DAE system (\ref{peel}) at stick slip points accompanied by the possible non-differentiability of $v$ due to 
near jump changes in the exponentially faster time scale makes almost all
DAE solvers (such as the Runge-Kutta solvers in MATLAB \footnote{http://www.mathworks.com/products/matlab}, the Gauss-Legendre Runge-Kutta methods 
and the Backward Differentiation Formula algorithm of the 
DDASPK \footnote{http://www.cs.ucsb.edu/$\sim$cse/software.html}) fail to integrate (\ref{peel}) directly as a DAE. The homogenized approach
reduces the stiffness and within an implicit solver the condition number of the Jacobian matrix in the non-linear solution phase
can be improved by a choice of appropriate $\tau$. Of course, larger the $\tau$, the more smoothed the solution of $v$ is. 
That is, $v$ converges only in the weak sense with respect to $\tau$. 
Hence it is important for the simulation to employ a numerical method that can self regularize its Jacobian matrix without losing stability so that 
the model given by (\ref{av}--\ref{str}) along with (\ref{hv}) can be studied with as small a $\tau$ as possible with a view to capturing the
stiff and oscillatory behavior of $v$ in the stick-slip regime. At the smallest $\tau$, the numerical method Jacobian approaches numerical rank \cite{golub}
deficiency, elucidating local bifurcations, if any, in the stick-slip regime.

In this work, we use the {\it $\alpha$-method} (described in detail in  \cite{GAlpha, Multi} and
not connected with the peel front angle $\alpha$) for the time integration of (\ref{av}--\ref{str}) appended with (\ref{hv}), i.e.,
the homogenized reformulation of (\ref{peel}). We define 
$$\psi_s := (f^{\mathsf T}_s, -J_s^{-1} \frac{\partial \big((1 + \frac{l \sin\alpha}{L(\alpha)}) F(u,\alpha) \big)_s}{\partial x_s} f_s)^{\mathsf T}:
{\mathbb R}^4 \to {\mathbb R}^4,~~y := (x^{\mathsf T}, v)^{\mathsf T} \in {\mathbb R}^4.$$ The {\it $\alpha$-method} numerically 
integrates (\ref{av}--\ref{str}, \ref{hv}) over $[t_0, t_f]$ using a uniform mesh of $N$ equal time steps, each  of 
size $(t_f - t_0)/N$. If the data at the $n$th discretization point is known, then the numerical solution at 
$(n+1)$th (indicated by the subscript) is found by the {\it $\alpha$-method} as follows.
\begin{subequations}
\begin{eqnarray}
y_{n+1}&=& y_n+\left(1-\frac{\beta}{\gamma}\right)h \psi(y_n) +\frac{\beta}{\gamma}h \psi(y_{n+1}) +\left(\frac{1}{2}-\frac{\beta}{\gamma}\right)h a_n\\
a_{n+1}&=&\frac{\psi(y_{n+1})-\psi(y_n)}{\gamma} +\left(1-\frac{1}{\gamma}\right) a_n \\
\gamma & := & \frac{2}{\rho+1}-\frac{1}{2};~~\beta:=\frac{1}{(\rho+1)^2},~~\rho \in [0,1)
\end{eqnarray}
\end{subequations}
where $h= t_{n+1}-t_n = (t_f - t_0)/N $ is the uniform time step size, and $a \in {\mathbb R}^4$ is an algorithmic variable that is merely updated
after every time step but does not need to be iteratively solved for in a time step within the nonlinear solver. 
The initial condition $a_0$ is calculated as $\frac{\mathrm{d}\psi}{\mathrm{d}t}$ at $t=0$.
The {\it $\alpha$-method} is an implicit $A$-stable integrator with the property of 
producing well-conditioned Jacobian matrices \cite{GAlpha, Multi} particularly when the algorithmic parameter $\rho$ is closer to zero.
The method is second order accurate in time step size for $y$ when (\ref{peel}) when $\psi$ is at least twice differentiable with respect to time. 
When $v$ is such that $\psi$ is only Lipschitz continuous then the {\it $\alpha$-method} is first order accurate in time step size. 
However, in the high differential index stick slip regime, $v$ may have poor smoothness leading to $\psi$ not being differentiable with respect to time
at some points. At these points the $\alpha$ method have an error linear in step size for $y$. The homogenized reformulation
(\ref{hv}) of (\ref{sac}) averages out the oscillatory and the stiff response in $v$ over a stretched time scale and thus increases the smoothness of 
$v$ making $\psi$ smooth almost everywhere over the time interval of simulation. 
A detailed error analysis of the {\it $\alpha$-method} under various continuity conditions may be found in \cite{GAlpha, Multi}.
\section{Numerical Example}
We consider an example with parameter values chosen following \cite{de2005missing, de2004dynamics, de2006dynamics}.
The peeling dynamics example is numerically simulated by the {\it $\alpha$-method} as the ODE system (\ref{av}--\ref{str}, \ref{hv}).
\begin{figure}[h,t,b]
 \centering
 \includegraphics[scale=0.59]{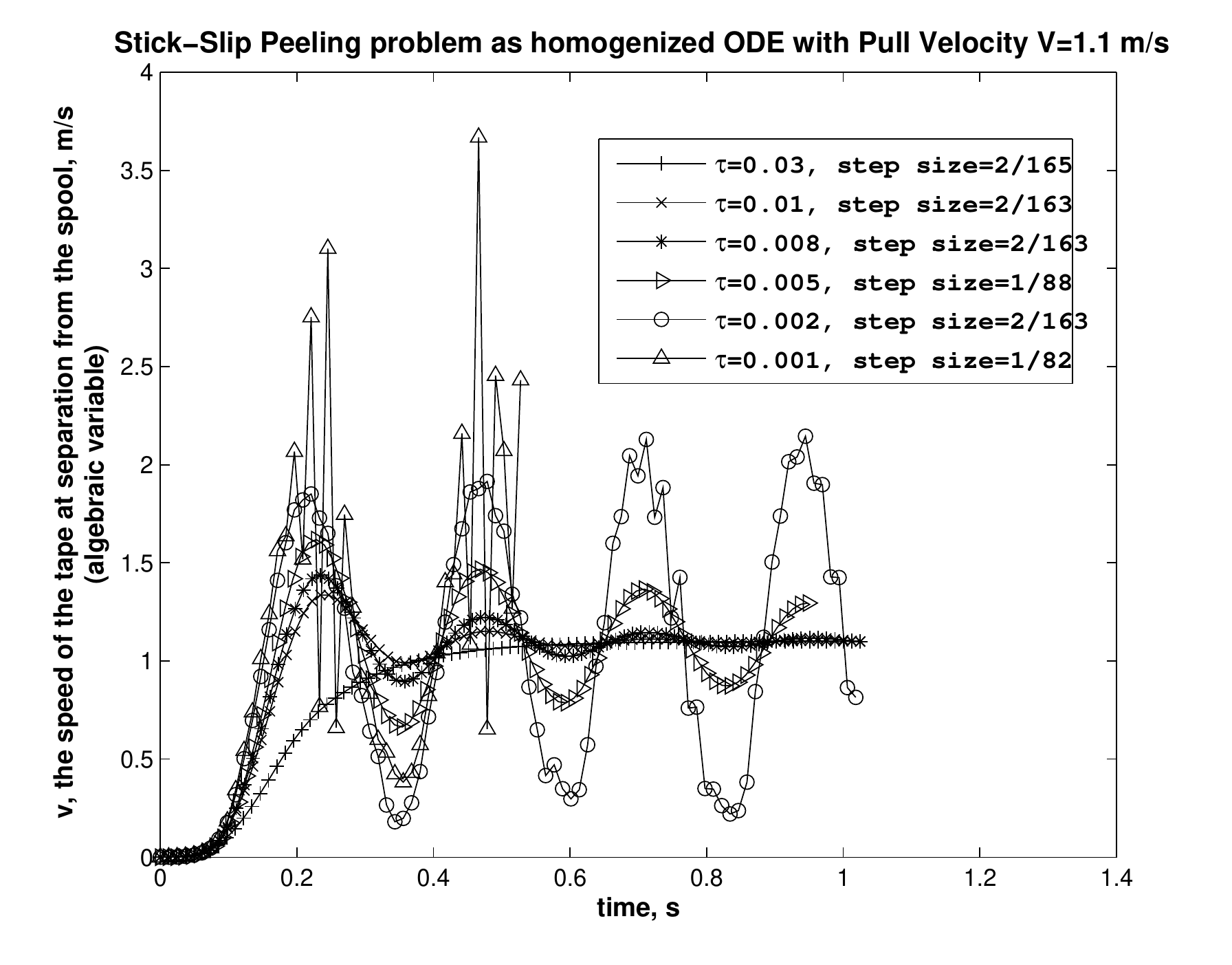}
\caption{Peeling velocity response with $V=1.1$ m/s. $\tau=0.001$ is the least $\tau$ for which 
the $\alpha$-method could be effective and $v$ shows stiff non-smooth changes and is oscillatory. A time step of $1/82$ s was the best accuracy that could be obtained without 
failure in convergence of the Newton iterations inside the  {\it $\alpha$-method}. It may be noted that at this pull velocity the local bifurcations happen often and $\tau$ is larger
than that for the other two pull velocities used in this work. Also, due to relatively 
ill-conditioned Jacobian at stick-slip, simulation stops early for a smaller $\tau$. 
As $\tau$ increases, the time profile of $v$ is smoothed to its average value which is the pull velocity.}\label{vODE1}
\end{figure}
\begin{figure}[h,t,b]
 \centering
 \includegraphics[scale=0.58]{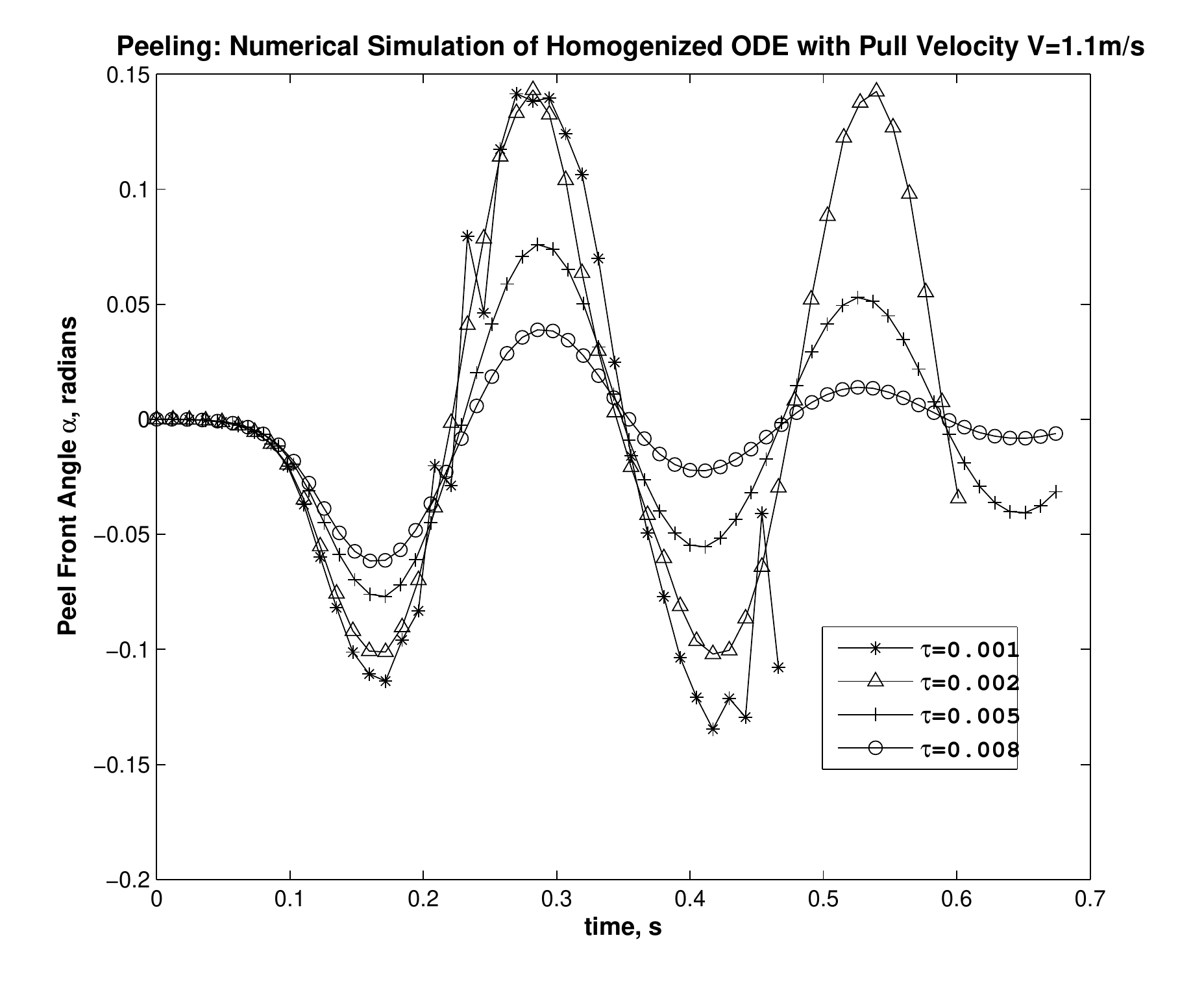}
\caption{Peeling front angle ($\alpha$) response with $V=1.1$ m/s. As $\tau$ increases, the time profile of $\alpha$ 
is smoothed to its average value which is zero.  At $\tau=0.001$ the slower scale compared to $v$ is seen.}\label{avODE1}
\end{figure}
\begin{figure}[h,t,b]
 \centering
 \includegraphics[scale=0.58]{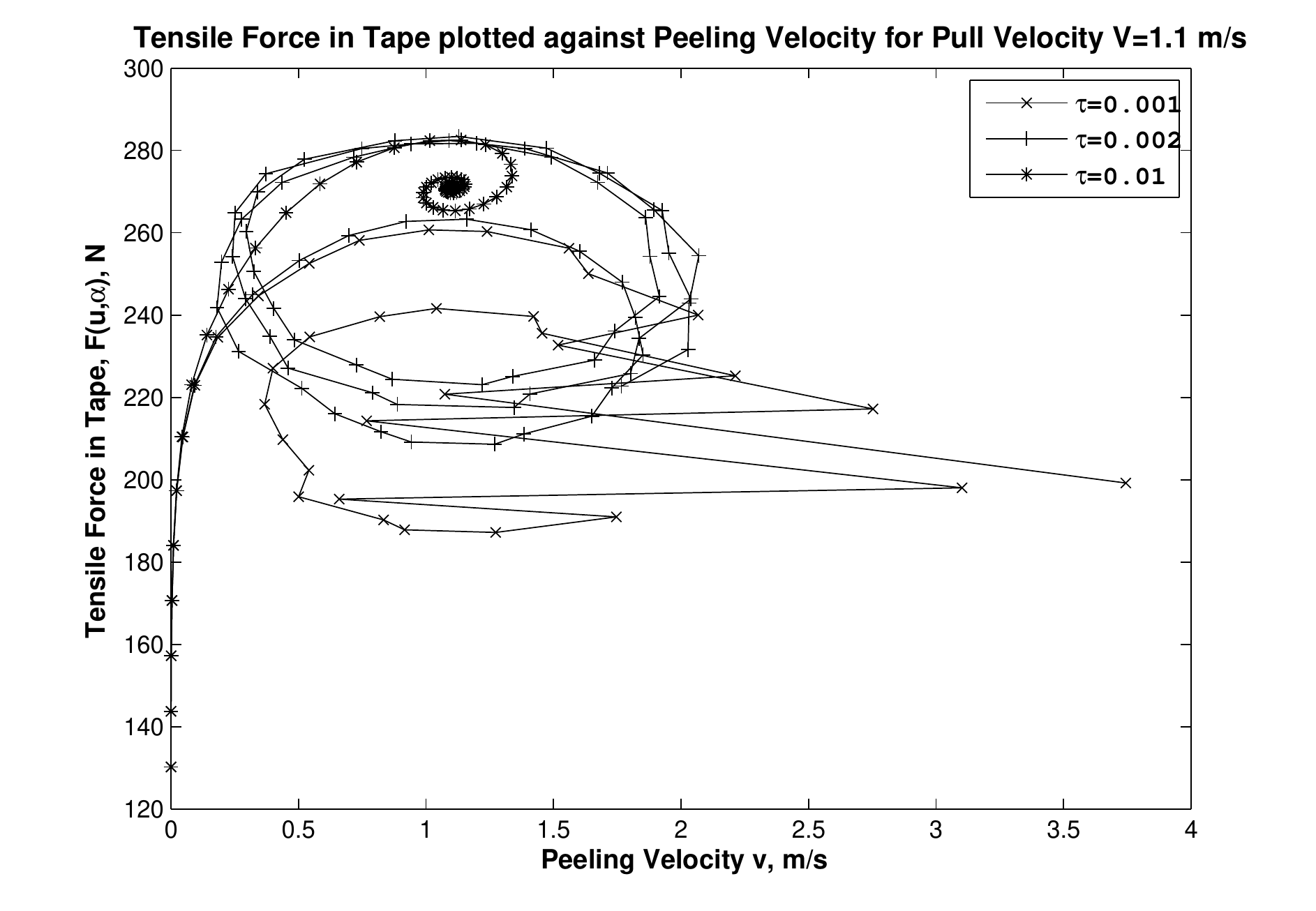}
\caption{Tensile force in Tape $F(u,\alpha)$  plotted against Peeling Velocity $v$ for Pull Velocity $V=1.1$ m/s. At low $\tau$ the bifurcations can be
seen as $v$ changes sharply corresponding to very little change in   $F(u,\alpha)$. This also demonstrates the very fast time scale of changes in $v$ 
and its stiffer response in bifurcations compared to $F(u,\alpha)$ which is a function of the slower differential variables. For a larger $\tau$ as $v$ averages
out to approach $V$, the graph goes to a fixed point at which $v=V$.} \label{FvODE1}
\end{figure}
\begin{figure}[h,t,b]
 \centering
 \includegraphics[scale=0.58]{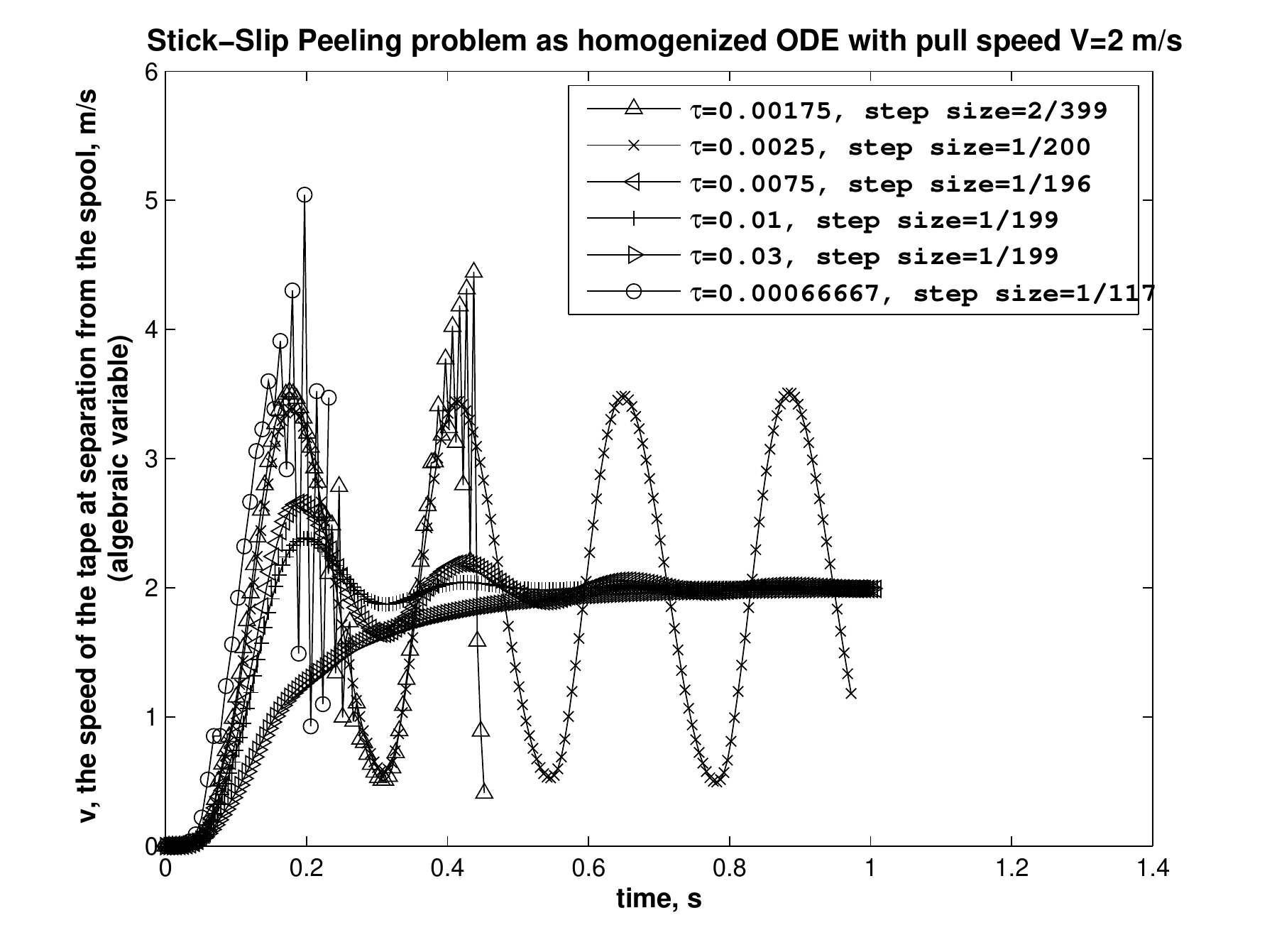}
\caption{Peeling velocity response with $V=2$ m/s. Here smaller time step size could be
taken, increasing the accuracy of the simulation. However, $\tau$ is of the same order as in $V=1.1$ as the stick-slip caused bifurcations remain significant.
The least $\tau$ has a slightly lower value as the differential variables get faster due to a higher pulling velocity $V$. 
As $\tau$ increases, $v$ is averaged out and approaches $V$. For a smaller $\tau$, the Jacobian is more ill-conditioned and the numerical 
integration stops early.}\label{vODE2}
\end{figure}
\begin{figure}[h,t,b]
 \centering
 \includegraphics[scale=0.58]{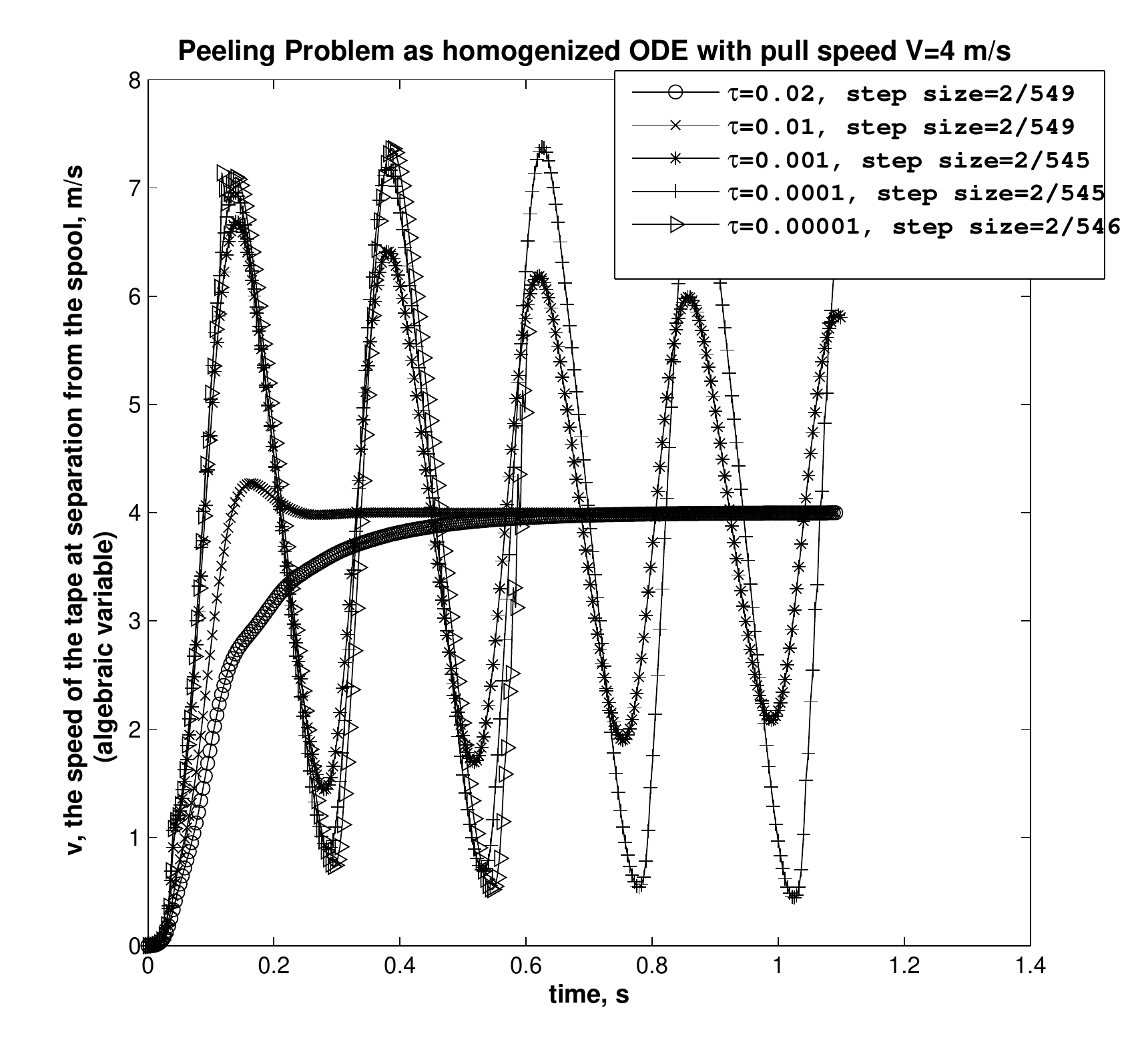}
\caption{Peeling velocity response with $V=4$ m/s. At this pull velocity there are fewer bifurcations in the stick-slip regime.
 The time profile of $v$ is also smoother
and less stiff. Under a higher $V$, the shear force peeling the tape is higher and the differential variables are faster, almost matching up with the time scale of $v$. 
Hence a much smaller time step size and much smaller $\tau$ could be used. 
As $\tau$ increases, $v$ is averaged out and approaches $V$. }\label{vODE3}
\end{figure}
\begin{figure}[h,t,b]
 \centering
 \includegraphics[scale=0.58]{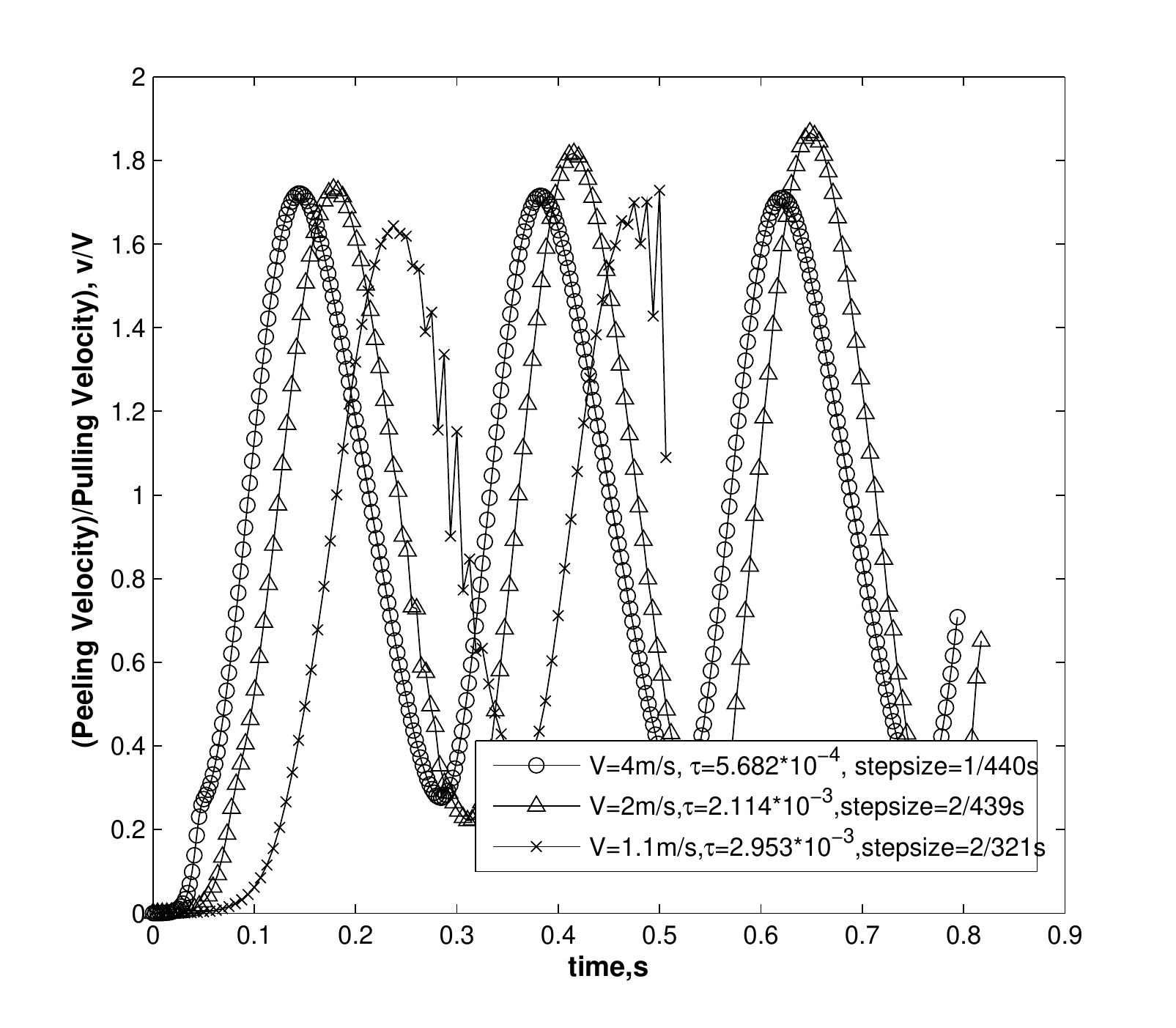}
\caption{$v$ as a solution of the homogenized ODE for various $V$'s. For a higher $V$ the problem could be solved with 
a smaller minimum $\tau$ for which the integrator Jacobian does not become non-invertible. Also a higher $V$ could admit 
a smaller time step size allowing more accuracy in the solution. This is because at a lower $V$ the fast scale of the stick-slip
regime dominates. For $V=1.1$ m/s the simulation stops early as a time step encounters relatively ill-conditioned Jacobian during
a sticking phase. Also, for this pull velocity the high frequency smaller amplitude oscillations happen during the sticking phase due
to the high stiffness and high differential index of the model.}
\label{vComp}
\end{figure}

The parameters are in SI units: $R=0.1,~I=10^{-2},~k/l=1000$ and for computational 
purposes, the approximations $L(\alpha) \approx l$, $l \gg u$ and $l \gg R$ hold.  The initial data for all the simulations are: $v_0 =10^{-10} \times V,
~\omega_0 = v_0/R,~ \alpha_0 =\frac{\pi}{4} \times 10^{-8}$ and the adhesion model \cite{de2006dynamics} is taken as $\phi(v,V):= 400 v^{0.35}+110 v^{0.15}+130 \exp{(v/11)}-2 V^{1.5}-
(415-45 V^{0.4}-0.35 V^{2.15})v^{0.5}$. The user-selectable parameter $\rho$ in the {\it $\alpha$-method} is set to zero for
keeping the integrator Jacobian matrix (cf.  \cite{GAlpha} for theoretical details) as well-conditioned as possible. 
The values of the constant pulling velocity, $V$, are chosen following the difficulty encountered in simulation of the 
peeling dynamics in \cite{de2005missing, de2004dynamics, de2006dynamics}. 
Since the peeling velocity $v$ is characterized by the local bifurcations and is the algebraic variable in the stick-slip dynamics, we study in Figures 
\ref{vODE1}, \ref{vODE2} and \ref{vODE3} the numerically computed time profiles of $v$ for various values of the pull velocity, $V$. 
Figure \ref{vODE1} corresponds to $V=1.1$ m/s, Figure \ref{vODE2} to $V=2$ m/s and Figure \ref{vODE3} to $V=4$ m/s. As the pulling velocity $V$ increases, 
the shear force overcoming the adhesion increases and the local bifurcations occur less frequently. 
For $V=4$ m/s, bifurcations are less accentuated and the time profile of $v$ is smoother
and less stiff as the spool rotates relatively unhindered by the adhesion. 
For $V=1.1$ and $2$ m/s, the stick slip regime dominates.
A smaller $\tau$ for $V=1.1, 2$ m/s would mean capturing the very stiff near-jump changes in $v$ during the stick-slip regime. 
The  actual computation is limited by the smallest real number the machine can represent and the numerical method can render 
the peeling velocity in the stick-slip regime only to the extent the integrator Jacobian remains numerically full rank (cf. \cite{golub} for numerical
rank of a matrix). 
However, for $V=4$ m/s the smoother and the less stiff response of $v$ allows one to use a smaller $\tau$. 
Thus the least $\tau$ for which the numerical integration can proceed without encountering a numerically rank deficient integrator Jacobian
is more at $V=1.1,2$ m/s than that at $V=4$ m/s. This illustrates that the parameter $\tau$ is essentially the time scale to which 
we are able to resolve and observe the stick-slip dynamics numerically. Lesser the bifurcations, computationally
it is easier to attain a finer resolution. Figure \ref{FvODE1} shows the bifurcations with respect to $F(u, \alpha)$ for the pull velocity $V=1.1$ m/s which
has the most dominant stick-slip regime. The faster changes in $v$ compared to $F(u, \alpha)$ may be noted in the same figure.
For all three values of $V$ considered, the increase in $\tau$ leads to coarser time resolution of $v$ and the 
largest $\tau$ resolves it only to a smoothed time profile which goes to the average value $v=V$. 
Figure \ref{avODE1} shows the time profile of $\alpha$ for $V=1.1$ m/s. As $\tau$ increases, $\alpha$ reaches its average value
which is zero and for the least $\tau$ it shows smaller amplitude faster oscillations along the slower and smoother trajectory. At $\tau=0.001$
it can be seen that $v$ is stiffer and more oscillatory than $\alpha$. Figure \ref{vComp} compares
the relative time profiles of $v$ at various values of $V$ and shows the difficulty in the numerical simulation 
when trying to capture the stick-slip regime behavior of the peeling velocity for the lower pull velocities which produce more frequent and pronounce stick-slips.
% \par
% In an experimental set-up the frequency with which observations can be made is limited by the equipment property. The observations of bifurcations 
% \cite{ciccotti13, cortet2007imaging} in terms of
% the peeling velocity and its resolution can be explained in terms of the parameter $\tau$. Noting the qualitative similarity between the time profiles of $v$ in 
% \cite{ciccotti13} and Figures \ref{vODE1} and \ref{vODE2}, we can conclude that an observation of $v$ with faster sampling rates essential capture the peeling
% dynamics for a smaller $\tau$.
\section{Conclusion}
We have shown that the bifurcations in the peeling dynamics of an adhesive tape are a consequence of the rank deficiency of the 
Jacobian and the high local differential index of the model which are structural properties of the DAE model of the peeling dynamics.
The bifurcations characterize the peeling velocity and also makes changes in the peeling velocity exponentially faster than 
the peel front angle and/or the tensile displacement of the tape. The homogenized ODE model presented in this work captures 
the characteristic time scale of the stick-slip dynamics 
by introducing the parameter $\tau$. This is important since a DAE cannot be studied as an ODE because of 
its inherent singularity and any ODE approximation, such as the homogenized ODE presented in 
this work, must have consistent convergence properties that smooth out the singularity
inherent in the DAE. At the smallest $\tau$ the homogenized ODE approximation approaches the 
DAE behavior since its Jacobian approaches rank deficiency. 
The numerical simulations corroborate the smoothing property of the homogenized ODE approach and 
at smaller values of $\tau$ elucidate the local bifurcations.
\section*{Acknowledgment}
The authors are grateful to Prof. G. Ananthakrishna of the Materials Research Centre of the Indian Institute of Science, Bangalore, India for sharing his insights
and for introducing them to the problem.

% \bibliographystyle{IEEEtran.bst}
% \bibliography{peelArxiv}
\end{document}